\documentclass[final,1p,times]{elsarticle}

\usepackage{lineno,hyperref,latexsym,amsmath,amsthm,amsfonts, bookmark}
\usepackage{listings}

\newtheorem{lem}{Lemma}
\newtheorem{thm}{Theorem}
\newtheorem{coro}{Corollary}
\newtheorem{rem}{Remark}
\newtheorem{prop}{Proposition}
\newtheorem{exa}{Example}

\newtheorem{con}{Conjecture}

\begin{document}
\newenvironment{prf}[1][Proof]{\noindent\textit{#1}\quad }

\begin{frontmatter}
\title{Dimensions of some LCD BCH codes}

\author[1]{Hanglong Zhang\corref{cor1}}
\ead{zhanghl9451@163.com}
\author[1,2]{Xiwang Cao}
\ead{xwcao@nuaa.edu.cn}
{\cortext[cor1]{Corresponding author.}}


\tnotetext[]
{This research is supported by the National Natural Science Foundation of China under Grant Nos. 11771007 and  12171241.}
\address[1]{Department of Mathematics, Nanjing University of Aeronautics and Astronautics, Nanjing 210016, China}
\address[2]{Key Laboratory of Mathematical Modelling and High Performance Computing of Air Vehicles (NUAA), MIIT,  Nanjing  211106,  China}

\begin{abstract}
 In this paper,  we  investigate the first few largest coset leaders modulo $\frac{q^m+1}{\lambda}$ where $\lambda\mid q+1$ and $q$ is an odd prime power, and give the dimensions of some LCD BCH codes of length $\frac{q^m+1}{\lambda}$ with  large designed distances.
We also determine the dimensions of  some LCD  BCH codes  of length $n=\frac{(q^m+1)}{\lambda}$ with  designed distances  $2\leq \delta \leq \frac{ q^{\lfloor(m+1)/2\rfloor}}{\lambda}+1$,  where $ \lambda\mid q+1$  and $1<\lambda<q+1$.  The LCD BCH codes presented in this paper have a sharper lower bound on the minimum distance than the BCH bound.
\end{abstract}


\end{frontmatter}

\section{Introduction}
 Through this paper, let $q$ be a prime power and $\mathbb{F}_q$ be the finite field with $q$ elements. An $[n,k,d]$ linear code $\mathcal{C}$ over $\mathbb{F}_q$ is a $k$-dimensional subspace of $\mathbb{F}^n_q$ with  minimum (Hamming) distance $d$.
 A linear  code $\mathcal{C}$  of length $n$ over $\mathbb{F}_q$ is called cyclic if for every $(c_0,c_1,\cdots,c_{n-1})\in\mathcal{C}$ implies $(c_{n-1},c_0,c_1,\cdots,c_{n-2})\in\mathcal{C}$. 
 Since the vector $(c_0,c_1,\cdots,c_{n-1})\in\mathbb{F}^n_q$ can be  identified with a  polynomial $c_0+c_1x+\cdots+c_{n-1}x^{n-1}\in \mathbb{F}_q[x]/\langle x^n-1 \rangle$,  a cyclic code of length $n$ over $\mathbb{F}_q$  is an ideal of the quotient ring $\mathbb{F}_q[x]/ \langle x^n-1 \rangle$. 
 Furthermore,  every ideal in  $\mathbb{F}_q[x]/ \langle x^n-1 \rangle$ is a  principal ideal, then $\mathcal{C}$ can be expressed as $\mathcal{C}=\langle g(x)\rangle$, where $g(x)$ is a divisor of $x^n-1$ and has the least degree.
 Then $g(x)$ is called the generator polynomial of $\mathcal{C}$ and $h(x)=(x^n-1)/g(x)$ is called the check polynomial of $\mathcal{C}$. Note that $x^n-1$ has no repeated factors over $\mathbb{F}_{q}$ if and only if ${\rm gcd}(n,q)=1$. In this paper, we always assume that ${\rm gcd}(n,q)=1$.

 Let $n$ be a positive integer and $m={\rm Ord}_n(q)$, that is,  $m$ is the smallest positive integer such that $n\mid q^m-1$. Let $\alpha$ be a generator of $\mathbb{F}^{\ast}_{q^m}$ and  $\beta=\alpha^{\frac{q^m-1}{n}}$. Then $\beta$ is a primitive $n$-th root of unity in $\mathbb{F}_{q^m}$.  It follows that $x^n-1=\prod_{i=0}^{n-1}(x-\beta^i)$. Let $m_i(x)$  denote the minimal polynomial of $\beta^i$ over $\mathbb{F}_q$.
 
 For any $2\leq \delta \leq n$,  we
 define $g_{(q,n,\delta,b)}(x)$ as  the least common multiple of $m_b(x)$, $m_{b+1}(x)$, $\cdots$, $m_{b+\delta-2}(x)$, in symbols, $$g_{(q,n,\delta,b)}(x)={\rm lcm}(m_b(x), m_{b+1}(x),\cdots,m_{b+\delta-2}(x)),$$ where $b$ is an integer.
 Let $\mathcal{C}_{(q,n,\delta,b)}$ denote the cyclic code of length $n$ with the generator polynomial $g_{(q,n,\delta,b)}(x)$. Then  $\mathcal{C}_{(q,n,\delta,b)}$ is called a BCH code. If $n=q^m-1$ , then $\mathcal{C}_{(q,q^m-1,\delta,b)}$ is called  a primitive BCH code. If $b=1$, then  $\mathcal{C}_{(q,n,\delta,1)}$ is called   a narrow-sense BCH code. If $n=q^m+1$,  then  $\mathcal{C}_{(q,q^m+1,\delta,b)}$ is called an antiprimitive BCH code by Ding \cite{DING}.
 Since $g_{(q,n,\delta,b)}(x)$ has $\delta-1$ consecutive roots $\beta^i$ for $b\leq i \leq b+\delta-2$, it follows from the BCH bound that the minimum distance of  $\mathcal{C}_{(q,n,\delta,b)}$ is at least $\delta$. Due to this fact, $\delta$ is called the designed distance of  $\mathcal{C}_{(q,n,\delta,b)}$. The  (Euclidean) dual code  of  a linear code  $\mathcal{C}$ of length $n$ over $\mathbb{F}_q$ is defined by $\mathcal{C}^{\perp}=\{ x\in \mathbb{F}_{q}^n |~ (x,y)=x\cdot y^T=0~ {\rm for~ all~} y\in \mathcal{C}\}$,
 where $y^T$ denotes the transpose of the vector $y=(y_1,\cdots,y_n)$. A linear code $\mathcal{C}$ over $\mathbb{F}_q$ is called an LCD code ( linear code with complementary dual ) if $\mathcal{C}\bigcap \mathcal{C^\perp}=\{0\}$.

 Cyclic codes are a subclass of linear codes and  widely employed for  
 storage devices, communication systems,  consumer electronics and other fields as they have a clear algebraic structure and efficient encoding and decoding algorithms, and thus have  been extensively studied (for example \cite{Char,LIU, LIU2,LIU3,SHI}). BCH codes are a significant class of cyclic codes.
 BCH codes were first introduced in 1960s by Hocquenghem \cite{Ho}, Bose and Ray-Chaudhuri \cite{Bose}. 
 The parameters of BCH codes have been studied in several works such as  \cite{Augot, Augot2, Canteaut, Charpin2,Ding2,ding3,DING,Kasami,Kasami2,Li,li2,liu,Liu2,liu3,wang,Yan, yue,yue1,yue2,shi} and references therein.
 However,    their parameters are unknown in general and   how to determine the parameters of a BCH code  is a  hard  problem as pointed out by Charpin \cite{Char} and Ding \cite{DING}.  
 
 The  research  of antiprimitive BCH codes is a hot topic in recent years. In \cite{Li,liu,Liu2,liu3,wang,Yan,shi}, the authors studied the coset leaders modulo $q^m+1$ and the parameters of antiprimitive BCH codes. In this paper,  we   determine  the    first few largest coset leaders modulo $\frac{q^m+1}{\lambda}$ where  $\lambda\mid q+1$ and $q$  is an odd prime power, and then  get the dimensions of some LCD  BCH codes of length $q^m+1$ with   large designed distances.
 We also determine  the dimensions of the LCD BCH codes $\mathcal{C}_{(q,(q^m+1)/\lambda,\delta+1,0)}$  with $1\leq \delta-1 \leq \frac{ q^{\lfloor(m+1)/2\rfloor}}{\lambda}$  and $\lambda\mid q+1$.
 Moreover,  these LCD BCH codes  have a sharper lower bound on the minimum distance than the well-known BCH bound.

 Here goes the structure of this paper. In Section 2, we  introduce some definitions and    useful results.
 In Section 3, we   determine the  first few largest coset leaders modulo $q^m+1$, where $q$  is an  odd prime power,  and  $m=4,8$, or $m\geq12$ and  $m\equiv 4~ ({\rm mod}~ 8)$. Based on these results,   the dimensions of some LCD BCH codes of length $q^m+1$ with  large designed distances are determined.  
 In Section 4,  for  $1\leq \delta-1 \leq \frac{ q^{\lfloor(m+1)/2\rfloor}}{\lambda}$,   the dimensions of  LCD BCH codes  $\mathcal{C}_{(q,n,\delta+1,0)}$  are determined, where $n=\frac{q^m+1}{\lambda}$ and  $1<\lambda\mid q+1$.  We also give the  first few largest
 coset leaders modulo $\frac{q^m+1}{\lambda}$ for some special values of $\lambda$   and thus determine the dimensions of  some LCD BCH codes of length $n$ with  large designed distances. 
 
\section{Preliminaries}\label{sec2}

In this section, we present some necessary background concerning cyclotomic cosets,  coset leaders, BCH codes and LCD codes.

Let $\mathbb{Z}_{n}$ denote the  ring of integers modulo $n$. Let $s$ be an integer with $0 \leq s < n$. The $q$-cyclotomic coset  of $s$ modulo $n$ is  defined by 
$$C_s=\{s,sq,sq^2,\cdots,sq^{l_s-1}\}~ ({\rm mod}~ n) \subset \mathbb{Z}_{n}, $$
where $l_s$ is the smallest positive integer such that $sq^{l_s}\equiv s ~({\rm mod}~ n)$.  Thus, $\big| C_s  \big|=l_s$. The smallest integer in $C_s$ is called the coset leader of $C_s$. Let $\varGamma_{(n,q)}$ be the set of  the coset leaders. Then we have $C_s\cap C_t=\emptyset$ for any distinct elements $s$ and $t$ in $\varGamma_{(n,q)}$, and $$  \bigcup_{s\in \varGamma_{(n,q)}} C_s =\mathbb{Z}_n.$$
That is, the $q$-cyclotomic cosets modulo $n$ form a  partition of $\mathbb{Z}_n$.  
It is known that the roots of $g_{(q,n,\delta,b)}(x)$ form the set $\{\beta^t:t\in T\}$, where $T$ is the union of some $q$-cyclotomic cosets, and is called the defining set of $\mathcal{C}_{(q,n,\delta,b)}$ with respect to $\beta$.
The dimension $k$ of $\mathcal{C}_{(q,n,\delta,b)}$ is equal to $n-|T|$.
Thus,  in   order  to determine the dimension of the BCH code  $\mathcal{C}_{(q,n,\delta,b)}$, we  need to determine the  coset leaders in  (or not in) the defining set $T$ and the sizes of the $q$-cyclotomic cosets modulo $n$ they are in.

LCD cyclic codes over finite fields were called reversible codes and studied by Massey \cite{Mas}. Let $f(x)=f_tx^t+f_{t-1}x^{t-1}+\cdots+f_1x+f_0$ be a polynomial over 
$\mathbb{F}_q$ with $f_t\neq0$ and $f_0\neq0$. Then the reciprocal polynomial of $f(x)$ is defined by $\widehat{f}(x)=f_0^{-1}x^tf(x^{-1})$. The following results about LCD cyclic codes over $\mathbb{F}_q$ are known in the literature \cite{yang}.
\begin{lem}\cite{yang}
	Let $\mathcal{C}$ be a cyclic code of length $n$ over $\mathbb{F}_q$ with the generator polynomial $g(x)$ and ${\rm gcd}(n,q)=1$.Then the following statements are equivalent:
	
	1. $\mathcal{C}$ is an LCD code.
	
	2. $g(x)$ is self-reciprocal, i.e., $g(x)=\widehat{g}(x)$.
	
	3. $\alpha^{-1}$ is a root of $g(x)$ for every root $\alpha$ of $g(x)$.
\end{lem}
It was shown in \cite{Li} that every cyclic code over $\mathbb{F}_q$ of length $n$ is an LCD code if $-1$ is a power of $q$ modulo $n$. Thus, it is easy to check that a cyclic code of length $n=(q^m+1)/\lambda$ is an LCD code, where $\lambda \mid q+1$.

\begin{lem}\cite{Li}\label{lem000}
	The BCH code  $\mathcal{C}_{(q,q^m+1,\delta,0)}$ has the minimum distance $d\geq 2(\delta-1)$.
\end{lem}
The above lemma gives a lower bound for the  minimum distance of some antiprimitive BCH codes, which is sharper than the well-known BCH bound.  By the same idea of lemma \ref{lem000}, we  get the following lemma and omit the proof.
\begin{lem}\label{lem8}
	Let $\lambda$ be a positive divisor of  $q+1$. Then	the BCH code  $\mathcal{C}_{(q,(q^m+1)/\lambda,\delta,0)}$ has the minimum distance $d\geq 2(\delta-1)$.
\end{lem}

The following two lemmas are well known, which can be used to determine the sizes of some cyclotomic cosets.
\begin{lem}\label{lem1}
	Let  $u,v$ and $b$ be  positive integers. Then
	\begin{eqnarray*}
		\operatorname{gcd}\left(b^{u}+1, b^{v}-1\right) & = & \left\{\begin{array}{ll}
			1, & \text { if } \frac{v}{\operatorname{gcd}(u, v)} \text { is odd and } b \text { is even, } \\
			2, & \text { if } \frac{v}{\operatorname{gcd}(u, v)} \text { is odd and } b \text { is odd, } \\
			b^{\operatorname{gcd}(u, v)}+1, & \text { if } \frac{v}{\operatorname{gcd}(u, v)} \text { is even. }
		\end{array}\right.
	\end{eqnarray*}
\end{lem}
For  a positive integer $b$, let $\nu_2(b)$ be the 2-adic order function of $b$, i.e., $\nu_2(b)$ is the largest integer  $l$ such that $2^l\mid b$ but $2^{l+1}\not\mid b$.
\begin{lem}\label{lem111}
 Let  $u,v$ and $b$ be  positive integers. Then
	\begin{eqnarray*}
		\operatorname{gcd}\left(b^{u}+1, b^{v}+1\right) & = & \left\{\begin{array}{ll}
			2, & \text { if } \nu_2(v)\neq\nu_2(u), \\
			b^{\operatorname{gcd}(u,v)}+1, & \text{if } \nu_2(v)=\nu_2(u) . 
		\end{array}\right.
	\end{eqnarray*}
\end{lem}

The following proposition was proved in \cite{shi}, which  gives a sufficient and necessary  condition for $0\leq s \leq q^m$ being a coset leader.
\begin{prop}\cite{shi}\label{prop1}
	Let $q$ be an odd prime power, $m\geq2$  a positive integer and $n=q^m+1$. Let $1\leq i\leq m-1$, $l$ and $h$ be integers satisfying 
	\begin{equation*}
		1 \leq l \leq \frac{q^{i}-1}{2} \text { and }-\frac{l\left(q^{m-i}-1\right)}{q^{i}+1}<h<\frac{l\left(q^{m-i}+1\right)}{q^{i}-1} \text {. }
	\end{equation*}
	Then $0\leq s\leq q^m$ is a coset leader if and only if  $s\leq \frac{n}{2}$ and $s\ne l q^{m-i}+h$.
\end{prop}

For any even positive integer, we have the following lemma.
%
%
\begin{lem}\label{lem00}
	Let $m$ be an even positive integer. Then  $m$ is a power of $2$, or  $m\equiv 2^{\nu_2(m)} ~({\rm mod}~ 2^{\nu_2(m)+1})$.	
\end{lem}
\begin{proof}
	Since $m$ is even,  $m=2^{\nu_2(m)}m'$ where $\nu_2(m)\geq 1$ and  $m'$ is odd. 
	Put $m'=1+2i$, where $i\geq 0$. If $i=0$, then $m=2^{\nu_2(m)}$. If $i>0$, then $m=2^{\nu_2(m)}+i\cdot2^{\nu_2(m)+1}$, i.e.,  $m\equiv 2^{\nu_2(m)}~ ({\rm mod}~ 2^{\nu_2(m)+1})$.
\end{proof}


\section{Dimensions of BCH codes of length \texorpdfstring {$n=q^m+1$}{}}
Assume that $q$ is an odd prime power and $m$ is an even positive integer. 
Denote  $\delta_{i}$ as the  $i$-th largest coset leader  modulo $n=q^m+1$.
In \cite{shi},  the authors proved that $\delta_{1}=\frac{n}{2}$  for any $m$. They also determined   $\delta_{2}$ and $\delta_{3}$ for $m=2$. 
In \cite{wang},  $\delta_{2}$ and $\delta_{3}$  were given for  $m\geq6$ and $m \equiv 2~ ({\rm mod~ } 4)$.	
In \cite{liu3}, for every even integer $m$, the authors determined the first largest and second largest coset leaders.
In order to state the main results in \cite{liu3}, we  define a function by 
\begin{eqnarray*}
	\Psi_q(x) & = & \left\{\begin{array}{ll}
		\frac{q-1}{2} \prod_{j=0}^{x}(q^{2^j}-1) , & \text { if } x\geq 0,  \\
		\frac{q-1}{2}	, & \text { if } x<0.
	\end{array}\right.
\end{eqnarray*}

\begin{thm}\cite{liu3,shi}\label{thm0}
	Let $q$ be an odd prime power. Then $\delta_{1}=\frac{n}{2}$  for any $m$, and 
	\begin{eqnarray*}
		\delta_{2} & = & \left\{\begin{array}{ll}
			\Psi_q(k-1) , & \text { if } m=2^k \text{~for~}k\geq 2, \\
			\frac{n}{q^{2^{\nu_2(m)}}+1}\Psi_q(\nu_2(m)-1), & \text { if } 
			m\equiv 2^{\nu_2(m)}~ ({\rm mod}~ 2^{\nu_2(m)+1}).
		\end{array}\right.
	\end{eqnarray*}
\end{thm}

Moreover, for the third largest and the fourth largest  coset leaders, a conjecture was left open in \cite{liu3}.
\begin{con}\label{con11} \cite{liu3}
	Let $q$ be an odd prime power. Then
	\begin{eqnarray*}
		\delta_{3} & = & \left\{\begin{array}{ll}
			\delta_2-2\Psi_q(k-3) , & \text { if } m=2^k \text{ for~}k\geq 2, \\
			\delta_2- 2\cdot\frac{\delta_2+\Psi_q(\nu_2(m))}{q^{2^{\nu_2(m)+1}}}, & \text { if } 
			m\equiv 2^{\nu_2(m)}~ ({\rm mod}~ 2^{\nu_2(m)+1}),
		\end{array}\right.
	\end{eqnarray*}	 
	and 
	\begin{eqnarray*}
		\delta_{4} & = & \left\{\begin{array}{ll}
			\delta_3-1, &\text {if } m=4, \\
			\delta_2- 2q^{2^{k-2}}\Psi_q(k-4), &\text {if } m=2^k \text{~for~}k\geq 3,\\
			\delta_3-2q^{2^{\nu_2(m)-1}}\Psi_q(\nu_2(m)-1), &\text{if~} m= 2^{\nu_2(m)}+ 2^{\nu_2(m)+1},\\
			\delta_3-2\Psi_q(\nu_2(m)) , &\text{if~} m= 2^{\nu_2(m)}+ t\cdot2^{\nu_2(m)+1}  \text{~for~}  t\geq2.
		\end{array}\right.
	\end{eqnarray*}		 
\end{con}
\begin{rem}
	For $m= 2^{\nu_2(m)}+ 2^{\nu_2(m)+1}$, $\delta_4$ is wrong in \cite{liu3}. Here, we give the corrected form of it.
\end{rem}
In this section,  we use  a different method from  \cite{liu3} to  determine $\delta_{2}$, $\delta_{3}$ and $\delta_{4}$ for $m=4,8$, or $m\geq12$ and  $m\equiv 4~ ({\rm mod}~ 8)$, respectively, which partially proves  Conjecture \ref{con11}.  

\subsection{\texorpdfstring {$m\geq 12$}{} and  \texorpdfstring{$m\equiv 4~ ({\rm mod}~ 8)$}{}}

In this subsection, we always assume that  $m\geq 12$ and  $m\equiv 4~ ({\rm mod}~ 8)$.
Denote 	
\begin{equation*}
\begin{aligned}	 
	& \Delta_2:=\frac{(q-1)^2(q^2-1)n}{2(q^4+1)}, \Delta_3:=\frac{(q-1)^2(q^2-1)}{2(q^4+1)}(q^m-2q^{m-8}-1) \text{ and }\\
	& \Delta_4:=	
	\begin{cases}
		\Delta_3-q^2(q^2-1)(q-1)^2, & m=12, \\
		\Delta_3-(q^4-1)(q^2-1)(q-1)^2, &m>12.\\
	\end{cases}
\end{aligned}
\end{equation*}

We aim to prove that $\delta_2=\Delta_2$, $\delta_3=\Delta_3$ and  $\delta_4=\Delta_4$, and  determine the sizes of $C_{\delta_2}$, $C_{\delta_3}$ and $C_{\delta_4}$. Consequently, the dimensions of  some BCH codes of length $q^m+1$  with  large designed distances are given.

Let  $\mathbb{Z}_m\backslash\{0\}=\{1,2,\cdots,m-1\}$. 	Let  $A_j:=\{11<i\leq m-1:  i \equiv j ~({\rm mod} ~8)\}$ for $0\leq j\leq7$. Let $R(\alpha,\beta)$ be the remainder of $\alpha$ dividing by $\beta$ and $R(\alpha,\beta)\geq 0$ for  positive integers $\alpha$ and $\beta$.
As is easily checked,
\begin{equation}\label{eq1}
\begin{aligned}
	\Delta_3&=\frac{(q-1)^2(q^2-1)}{2(q^4+1)}(q^m-2q^{m-8}-1) \\
	&=\frac{n}{2}-q^{m-1}+q^{m-3}-q^{m-4}+q^{m-5}-q^{m-7}+q^{m-9}-q^{m-11}\\
	&	+ \sum^{(m-12)/8}_{j=1}  (q^{8j}-q^{8j-1}+q^{8j-3}-q^{8j-4}+q^{8j-5}-q^{8j-7})\\
	&=\frac{n}{2}-q^{m-1}+q^{m-3}-q^{m-4}+q^{m-5}-q^{m-7}+q^{m-9}-q^{m-11}+ \sum^{(m-12)/8}_{j=1} \Phi(j),
\end{aligned} 
\end{equation}where $\Phi(j)=q^{8j}-q^{8j-1}+q^{8j-3}-q^{8j-4}+q^{8j-5}-q^{8j-7}$.
Note that if $(m-12)/8=0$, i.e. $m=12$, then $\sum^{(m-12)/8}_{j=1}  \Phi(j)=0$.

The following lemma divides $\mathbb{Z}_m\backslash\{0\}$ into two subsets $S_1$ and $S_2$, based on  whether $R(\Delta_3,q^{m-i})$ $\geq\frac{q^{m-i}+1}{2}$ or not, where $i$ runs through  $\mathbb{Z}_m\backslash\{0\}$.
\begin{lem}\label{lem3}
Let 	\begin{align*}
	&S_1=\{1,2,4,7,8,11\} \cup A_0 \cup A_3\cup A_5 \cup A_6  ~ {\rm and}~ \\
	&	S_2=\{3,5,6,9,10\} \cup A_1 \cup A_2\cup A_4 \cup A_7.
\end{align*} Then $S_1$ and $S_2$ form a partition of $\mathbb{Z}_m\backslash\{0\}$. Moreover, for each $i\in S_1$, $R(\Delta_3,q^{m-i})$ $\geq \frac{q^{m-i}+1}{2}$; for each $i\in S_2$, $R(\Delta_3,q^{m-i})< \frac{q^{m-i}+1}{2}$. 
\end{lem}
\begin{proof} For $0\leq \alpha,\beta\leq m-1$, $R(q^{\alpha},q^{\beta})=0$ if $\alpha\geq \beta$, and $R(q^{\alpha},q^{\beta})= q^{\alpha}$ if $\alpha< \beta$.   Since $\frac{n}{2}=\frac{q^m+1}{2}=\frac{q^{m-i}(q^i-1)+q^{m-i}+1}{2}$, we get that  $R(\frac{n}{2},q^{m-i})=\frac{q^{m-i}+1}{2}$.  As is easily checked,
\begin{equation}\label{eq0}
	\begin{aligned}
		R(\Delta_3,q^{m-i})&= R(\frac{n}{2},q^{m-i})-R(q^{m-1},q^{m-i})+R(q^{m-3},q^{m-i})-R(q^{m-4},q^{m-i})\\
		&+R(q^{m-5},q^{m-i})-R(q^{m-7},q^{m-i})+R(q^{m-9},q^{m-i})-R(q^{m-11},q^{m-i})\\
		&+ 	 \sum^{(m-12)/8}_{j=1}R(q^{8j},q^{m-i})-R(q^{8j-1},q^{m-i})+R(q^{8j-3},q^{m-i})\\
		&-R(q^{8j-4},q^{m-i})+R(q^{8j-5},q^{m-i})-R(q^{8j-7},q^{m-i}),
	\end{aligned}
\end{equation}
because the left-hand side is congruent to the right-hand side of Eq.(\ref{eq0}), and the absolute value of   sum of all the terms on the right-hand side of Eq.(\ref{eq0}), except $R(\frac{n}{2},q^{m-i})$, is less than $1+q+\cdots+q^{m-i-1}=\frac{q^{m-i}-1}{q-1}\leq \frac{q^{m-i}-1}{2}$, which implies  that the right-hand side of Eq.(\ref{eq0}) is greater than 0 and less than $q^{m-i}$.
Below  we determine which  values of $i$ result in the second term of $R(\Delta_3,q^{m-i})$ being negative, or not.
Suppose that $1\leq i\leq 11$. Then $m-11\leq m-i\leq m-1$.
When 
$ m-i=m-3,m-5,m-6,m-9$ and $m-10$, by Eq.(\ref{eq0}), the  second terms of $R(\Delta_3,q^{m-i})$ are $-q^{m-4},-q^{m-7},-q^{m-7},-q^{m-11}$ and $-q^{m-11}$, respectively. This implies that $R(\Delta_3,q^{m-i})<R(\frac{n}{2},q^{m-i})= \frac{q^{m-i}+1}{2}$ for $i\in\{3,5,6,9,10 \}$.
When $m-i=m-1,m-2,m-4,m-7,m-8$ and $m-11$, by Eq.(\ref{eq0}),  the  second terms of $R(\Delta_3,q^{m-i})$ are $q^{m-3},q^{m-3},q^{m-5},q^{m-9},q^{m-9}$ and $q^{m-12}$, respectively. This implies that $R(\Delta_3,q^{m-i})\geq  R(\frac{n}{2},q^{m-i})= \frac{q^{m-i}+1}{2}$ for $i\in\{1,2,4,7,8,11 \}$.

Suppose that $11<i\leq m-1$. Then $1\leq m-i \leq m-12$. When $m-i=8j$ for some $j\geq 1$, since   $m\geq 12$ and  $m\equiv 4~ ({\rm mod}~ 8)$, we have $i\equiv 4~ ({\rm mod}~ 8)$, i.e., $i\in A_4$. In this case, by Eq.(\ref{eq0}), the second term of $R(\Delta_3,q^{m-i})$ is $-q^{8j-1}$. This implies $R(\Delta_3,q^{m-i})<\frac{q^{m-i}+1}{2}$. With a same analysis, when $m-i\in\{8j-3,8j-5,8j-6\}$ for some $j\geq1$, i.e., $i\in A_7\cup A_1 \cup A_2$, the  second term of $R(\Delta_3,q^{m-i})$ is	 negative. This implies $R(\Delta_3,q^{m-i})< \frac{q^{m-i}+1}{2}$. When $m-i\in\{8j-1,8j-2,8j-4,8j-7\}$ for some $j\geq1$, i.e., $i\in A_5\cup A_6 \cup A_0\cup A_3$, the  second term of $R(\Delta_3,q^{m-i})$ is not negative. This implies $R(\Delta_3,q^{m-i})\geq \frac{q^{m-i}+1}{2}$. It is evident that  $S_1$ and $S_2$ form a partition of $\mathbb{Z}_m\backslash\{0\}$. This completes the proof.
\end{proof}

\begin{lem}\label{lem11}
With the notations given above. For  $i\in S_1$, let $h_1(i)=R(\Delta_3,q^{m-i})$ $-q^{m-i}$ and $l_1(i)=\frac{\Delta_3-h_1(i)}{q^{m-i}}$.  For $i\in S_2$, let $h_2(i)=R(\Delta_3,q^{m-i})$ and $l_2(i)=\frac{\Delta_3-h_2(i)}{q^{m-i}}$. Then
\begin{equation*}
	l_1(i)(q^{m-i}-1)<-h_1(i)(q^i+1) ~and~ l_2(i)(q^{m-i}+1)<h_2(i)(q^i-1).
\end{equation*}
\end{lem}
\begin{proof}
Suppose that $i\in S_1$. Let $i=7$. Then we have
\begin{align*}
	h_1(7)&=R(\Delta_3,q^{m-7})-q^{m-7}=\frac{q^{m-7}+1}{2}+q^{m-9}-q^{m-11}+ \sum^{(m-12)/8}_{j=1} \Phi(j)-q^{m-7},\\
	\Delta_3-h_1(7)&=\frac{q^m}{2}-q^{m-1}+q^{m-3}-q^{m-4}+q^{m-5}-\frac{q^{m-7}}{2}.
\end{align*}
By calculating only the first 3 terms in $	-h_1(7)(q^7+1)$,
\begin{align*}
	-h_1(7)(q^7+1)
	&=\big(\frac{q^{m-7}-1}{2}-q^{m-9}+q^{m-11}- \sum^{(m-12)/8}_{j=1} \Phi(j)\big)(q^7+1)\\
	&=\frac{q^m-q^7}{2}-q^{m-2}+q^{m-4}+\cdots,
\end{align*}  we  find that $l_1(7)(q^{m-7}-1)=\frac{\Delta_3-h_1(7)}{q^{m-7}}(q^{m-7}-1)< \Delta_3-h_1(7)=\frac{q^m}{2}-q^{m-1}+q^{m-3}-q^{m-4}+q^{m-5}-\frac{q^{m-7}}{2}<\frac{q^m-q^7}{2}-q^{m-2}<	-h_1(7)(q^7+1)$. With a same analysis, when $i\in\{1,2,4,8,11\}$, we can prove that $l_1(i)(q^{m-i}-1)< \Delta_3-h_1(i)<	-h_1(i)(q^i+1)$.  Let $i\in A_5$. Then $m-i=8j-1<m-11$ for some $j\geq 1$. Then we have 
\begin{align*}
	h_1(i)&=R(\Delta_3,q^{m-i})-q^{m-i}
	=\frac{-q^{m-i}+1}{2}+q^{m-i-2}-q^{m-i-3}+q^{m-i-4}-q^{m-i-6}+\sum_{j=1}^{(m-i-7)/8}\Phi(j).
\end{align*} 	By calculating only the first 4 terms in $ \Delta_3-h_1(i)$ and  $-h_1(i)(q^i+1)$,
\begin{align*}
	\Delta_3-h_1(i)&=\frac{n}{2}-q^{m-1}+q^{m-3}-q^{m-4}+\cdots,\\
	-h_1(i)(q^i+1)&=\frac{q^m-q^i}{2}-q^{m-2}+q^{m-3}-q^{m-4}+\cdots,
\end{align*} we can find that $l_1(i)(q^{m-i}-1)<\Delta_3-h_1(i)<   \frac{n}{2}-q^{m-1}+q^{m-3} < \frac{q^m-q^i}{2}-q^{m-2} <-h_1(i)(q^i+1)$. With a same analysis, when $i\in  A_0\cup A_3 \cup A_6$, we can prove that  $l_1(i)(q^{m-i}-1)<\Delta_3-h_1(i)<-h_1(i)(q^i+1)$.

Suppose that $i\in S_2$. Let $i=3$. Then we have 
\begin{align*}
	h_2(3)&=R(\Delta_3,q^{m-3})
	=\frac{q^{m-3}+1}{2}-q^{m-4}+q^{m-5}-q^{m-7}+q^{m-9}-q^{m-11}
	+ \sum^{(m-12)/8}_{j=1} \Phi(j),\\
	l_2(3) &=\frac{\Delta_3-h_2(3)}{q^{m-3}}=\frac{\frac{q^m}{2}-q^{m-1}+\frac{q^{m-3}}{2}}{q^{m-3}}=\frac{q^3+1}{2}-q^2,\\
	l_2(3)(q^{m-3}+1)&=\frac{q^m+q^{m-3}}{2}-q^{m-1}+\frac{q^3+1}{2}-q^2.
\end{align*} By calculating only the first 5 terms in $	h_2(3)(q^3-1)$,
\begin{align*}
	h_2(3)(q^3-1)&=\frac{q^m+q^3}{2}-q^{m-1}+q^{m-2}-q^{m-4}+q^{m-6}-\cdots,
\end{align*}	 we can find that $l_2(3)(q^{m-3}+1) <\frac{q^m+q^3}{2}-q^{m-1}+q^{m-2}-q^{m-4}<h_2(3)(q^3-1)$. With a same analysis, when $i\in\{5,6,9,10\}$, we can prove that $l_2(i)(q^{m-i}+1)<h_2(i)(q^i-1)$.  Let $i\in A_7$. Then $m-i=8j-3<m-11$ for some $j\geq 1$. Then we have 
\begin{align*}
	h_2(i)&=\frac{q^{m-i}+1}{2}-q^{m-i-1}+q^{m-i-2}-q^{m-i-4}+\sum_{j=1}^{(m-i-5)/8}\Phi(j),\\
	l_2(i)&=\frac{\Delta_3-h_2(i)}{q^{m-i}}=\bigg( \frac{q^m+q^{m-i}}{2}-q^{m-1}+q^{m-3}-q^{m-4}+q^{m-5}-q^{m-7}+q^{m-9}-q^{m-11}\\
&	+ \sum^{(m-12)/8}_{j=(m-i-5)/8+1} \Phi(j) +q^{m-i+3}-q^{m-i+2} \bigg)/q^{m-i}\\
	&=\frac{q^i+1}{2}-q^{i-1}+q^{i-3}-q^{i-4}+q^{i-5}-q^{i-7}+q^{i-9}-q^{i-11} + \sum^{(m-12)/8}_{j=(m-i-5)/8+1} \frac{\Phi(j)}{q^{m-i}}+q^3+q^2,
\end{align*} where $\sum^{(m-12)/8}_{j=(m-i-5)/8+1} \frac{\Phi(j)}{q^{m-i}}=0$ if $(m-12)/8<(m-i-5)/8+1$. 	By calculating only the first 4 terms in $l_2(i)(q^{m-i}+1)$ and   $ h_2(i)(q^i-1)$,
\begin{align*}
	l_2(i)(q^{m-i}+1)&=\frac{q^m+q^{m-i}}{2}-q^{m-1}+q^{m-3}-q^{m-4}+\cdots,\\
	h_2(i)(q^i-1)&=\frac{q^m+q^i}{2}-q^{m-1}+q^{m-2}-q^{m-4}+\cdots,
\end{align*} we  find that $l_2(i)(q^{m-i}+1)< \frac{q^m+q^{m-i}}{2}-q^{m-1}+q^{m-3}<  \frac{q^m+q^i}{2}-q^{m-1}+q^{m-2}-q^{m-4}<h_2(i)(q^i-1)$.
When $i\in A_1\cup A_2\cup A_4$, we can prove that $l_2(i)(q^{m-i}+1)< h_2(i)(q^i-1)$. This completes the proof.
\end{proof}

\begin{lem}\label{lem14}
Let  $q$ be an odd prime power, $m\geq12$ and $m\equiv 4~ ({\rm mod}~ 8)$ and $n=q^m+1$. Then $\Delta_3=\frac{(q-1)^2(q^2-1)}{2(q^4+1)}(q^m-2q^{m-8}-1)$ is a coset leader.
\end{lem}
\begin{proof}
From Proposition \ref{prop1}, $\Delta_3$ is a coset leader if and only if for each $i\in\mathbb{Z}_m\backslash\{0\}$, $l(i)$ and $h(i)$ are integers satisfying   $ 1\leq l(i) \leq \frac{q^{i}-1}{2}$  and $-\frac{l\left(q^{m-i}-1\right)}{q^{i}+1}<h(i)<\frac{l\left(q^{m-i}+1\right)}{q^{i}-1} $, then 
$	\Delta_3\neq l(i)q^{m-i} + h(i)$.
It suffices to  prove that if there exist  $i\in \mathbb{Z}_m\backslash\{0\}$,  $l(i)$ and $h(i)$ satisfy   $1\leq l(i) \leq \frac{q^{i}-1}{2}$, $-\frac{l\left(q^{m-i}-1\right)}{q^{i}+1}<h(i)<\frac{l\left(q^{m-i}+1\right)}{q^{i}-1}$ and 
\begin{equation}\label{equ1}
	\Delta_3= l(i)q^{m-i} + h(i),
\end{equation}
then $-\frac{l\left(q^{m-i}-1\right)}{q^{i}+1}>h(i)$ or  $h(i)>\frac{l\left(q^{m-i}+1\right)}{q^{i}-1} $.  

By the range of the values of $l(i)$ and $h(i)$, we  get $|h(i)|< \frac{q^{m-i}+1}{2}$.		From Eq.   (\ref{equ1}), one has $	h(i) \equiv \Delta_3  ~({\rm mod} ~q^{m-i})$, which is equivalent to 
\begin{equation}\label{eq2}
	\begin{aligned}
		h(i) &\equiv R(\Delta_3,q^{m-i})  ~({\rm mod} ~q^{m-i}). 
	\end{aligned} 
\end{equation} 
By Lemma \ref{lem3}, $\mathbb{Z}_m\backslash\{0\}$ can be divided into two disjoint subsets $S_1$ and $S_2$. Suppose that   $i\in S_1$, since  $R(\Delta_3,q^{m-i})\geq \frac{q^{m-i}+1}{2}$, $|h(i)|< \frac{q^{m-i}+1}{2}$ and (\ref{eq2}), we get $h(i)= R(\Delta_3,q^{m-i})-q^{m-i}$. Consequently, by Lemma \ref{lem11}, one has  $	 l(i)(q^{m-i}-1)<-h(i)(q^i+1)$, i.e., $-\frac{l(i)\left(q^{m-i}-1\right)}{q^{i}+1}>h(i)$.
Suppose that   $i\in S_2$, since  $R(\Delta_3,q^{m-i})<  \frac{q^{m-i}+1}{2}$, $|h(i)|< \frac{q^{m-i}+1}{2}$ and (\ref{eq2}), we get $h(i)= R(\Delta_3,q^{m-i})$. Consequently, by Lemma \ref{lem11}, one has $	 l(i)(q^{m-i}+1)<h(i)(q^i-1)$, i.e., $\frac{l(i)\left(q^{m-i}+1\right)}{q^{i}-1}<h(i)$. This completes the proof.
\end{proof}

Next, we prove that $\Delta_2$ and $\Delta_4$ are coset leaders.    As is easily checked, we have
\begin{equation*}
\begin{aligned}	 
	\Delta_2&=\frac{(q-1)^2(q^2-1)n}{2(q^4+1)}\\
	&=\frac{q^m+1}{2}-q^{m-1}+q^{m-3}
	+ \sum^{(m-12)/8}_{j=1}  -(q^{8j}-q^{8j-1}+q^{8j-3}-q^{8j-4}+q^{8j-5}-q^{8j-7})-1
\end{aligned}
\end{equation*}
and when $m=12$, we have 
\begin{align*}
\Delta_4&=\Delta_3-q^2(q^2-1)(q-1)^2
=\frac{q^m+1}{2}-q^{11}+q^9-q^8+q^7-q^6+q^5-q^3+q^2-q;
\end{align*}
when $m\geq20$, we have 
\begin{align*}
\Delta_4&=\Delta_3-(q^4-1)(q^2-1)(q-1)^2\\
&=\frac{q^m+1}{2}-q^{m-1}+q^{m-3}-q^{m-4}+q^{m-5}-q^{m-7}+q^{m-9}-q^{m-11}\\
&+\sum^{(m-12)/8}_{j=2}  (q^{8j}-q^{8j-1}+q^{8j-3}-q^{8j-4}+q^{8j-5}-q^{8j-7})\\
&+ q^7-q^5+q^4-q^3+q-1.
\end{align*}
The proofs of following, namely, Lemmas 10-15, are essentially the same as that given in the proofs of Lemmas \ref{lem3}, \ref{lem11} and \ref{lem14}, respectively, and are therefore omitted.

\begin{lem}
Let  $A_j:=\{3<i\leq m-1:  i \equiv j ~({\rm mod} ~8)\}$ for $0\leq j\leq7$.	Let 	\begin{align*}
	&S_1=\{1,2\}\cup A_1 \cup A_2\cup A_4 \cup A_7   ~and~ \\
	&	S_2=\{3\} \cup A_0 \cup A_3\cup A_5 \cup A_6.
\end{align*} Then $S_1$ and $S_2$ form a partition of $\mathbb{Z}_m\backslash\{0\}$. Moreover, for each $i\in S_1$, $R(\Delta_2,q^{m-i})$ $\geq \frac{q^{m-i}+1}{2}$; for each $i\in S_2$, $R(\Delta_2,q^{m-i})< \frac{q^{m-i}+1}{2}$. 
\end{lem}
\begin{lem}
With the notations given above. For  $i\in S_1$, let $h_1(i)=R(\Delta_2,q^{m-i})$ $-q^{m-i}$ and $l_1(i)=\frac{\Delta_2-h_1(i)}{q^{m-i}}$.  For $i\in S_2$, let $h_2(i)=R(\Delta_2,q^{m-i})$ and $l_2(i)=\frac{\Delta_2-h_2(i)}{q^{m-i}}$. Then
\begin{equation*}
	l_1(i)(q^{m-i}-1)<-h_1(i)(q^i+1) ~and~ l_2(i)(q^{m-i}+1)<h_2(i)(q^i-1).
\end{equation*}
\end{lem}
\begin{lem} \label{lem15}
Let  $q$ be an odd prime power, $m\geq12$ and $m\equiv 4~ ({\rm mod}~ 8)$ and $n=q^m+1$. Then $\Delta_2$ is a coset leader.
\end{lem}
\begin{lem}	Let  $A_j:=\{11<i<m-7:  i \equiv j ~({\rm mod} ~8)\}$ for $0\leq j\leq7$. When $m=12$, let 
\begin{align*}
	& S_1=\{1,2,4,6,9,11\} ~and~\\
	&S_2=\{3,5,7,8,10\};
\end{align*}
when $m\geq20$, let	
\begin{align*}
	&S_1=\{1,2,4,7,8,11,m-5,m-3,m-2\} \cup A_0 \cup A_3\cup A_5 \cup A_6  ~and~ \\
	&	S_2=\{3,5,6,9,10,m-7,m-6,m-4,m-1\} \cup A_1 \cup A_2\cup A_4 \cup A_7.
\end{align*} Then $S_1$ and $S_2$ form a partition of $\mathbb{Z}_m\backslash\{0\}$. Moreover, for each $i\in S_1$, $R(\Delta_4,q^{m-i})$ $\geq \frac{q^{m-i}+1}{2}$; for each $i\in S_2$, $R(\Delta_4,q^{m-i})< \frac{q^{m-i}+1}{2}$. 
\end{lem}
\begin{lem}
With the notations given above. For  $i\in S_1$, let $h_1(i)=R(\Delta_4,q^{m-i})$ $-q^{m-i}$ and $l_1(i)=\frac{\Delta_4-h_1(i)}{q^{m-i}}$.  For $i\in S_2$, let $h_2(i)=R(\Delta_4,q^{m-i})$ and $l_2(i)=\frac{\Delta_4-h_2(i)}{q^{m-i}}$. Then
\begin{equation*}
	l_1(i)(q^{m-i}-1)<-h_1(i)(q^i+1) ~and~ l_2(i)(q^{m-i}+1)<h_2(i)(q^i-1).
\end{equation*}
\end{lem}
\begin{lem} \label{lem16}
Let  $q$ be an odd prime power, $m\geq12$ and $m\equiv 4~ ({\rm mod}~ 8)$ and $n=q^m+1$. Then $\Delta_4$ is a coset leader.
\end{lem}

\begin{lem}\label{lem4}
With the notations given  before. Then 
\begin{equation*}
	\delta_2=\Delta_2,		\delta_3=\Delta_3 \text{ and } 	\delta_4=\Delta_4.
\end{equation*} 
\end{lem}
\begin{proof}
Let $i=1$ and $l=\frac{q-1}{2}$ in Proposition \ref{prop1}. Then $s$ is not a coset leader if 
\begin{equation*}
	\frac{(q-1)n}{2(q+1)}<s<\frac{n}{2}.
\end{equation*}
Let $i=2$ and $l=\frac{(q-1)^2}{2}$ in Proposition \ref{prop1}. Then $s$ is not a coset leader if 
\begin{equation*}
	\frac{(q-1)^2n}{2(q^2+1)}<s<\frac{(q-1)n}{2(q+1)}.
\end{equation*}
Let $i=4$ and $l=\frac{(q-1)^2(q^2-1)}{2}$ in Proposition \ref{prop1}. Then $s$ is not a coset leader if 
\begin{equation*}
	\frac{(q-1)^2(q^2-1)n}{2(q^4+1)}<s<	\frac{(q-1)^2n}{2(q^2+1)}.
\end{equation*}
Since $m$ is even, 	both $\frac{(q-1)^2n}{2(q^2+1)}$ and $	\frac{(q-1)n}{2(q+1)}$ are not integers.  Thus, for $	\Delta_2=\frac{(q-1)^2(q^2-1)n}{2(q^4+1)}<s<\frac{n}{2}=\delta_1$, $s$ is not a coset leader. By Lemma \ref{lem15}, $
\delta_2=\Delta_2$. 

Next, we prove that $		\delta_3=\Delta_3$. We aim to prove that for $0< u <  \Delta_2-\Delta_3$, $\Delta_{3}+u$ is not a coset leader.
By Proposition \ref{prop1}, we need to  prove that there always exists $i\in \mathbb{Z}_{m}\backslash\{0\}$,  $ 1\leq l(i) \leq \frac{q^{i}-1}{2}$  and $-\frac{l(i)\left(q^{m-i}-1\right)}{q^{i}+1}<h(i)<\frac{l(i)\left(q^{m-i}+1\right)}{q^{i}-1} $ such that 
\begin{equation}\
	\Delta_{3}+u=l(i)q^{m-i} + h(i).
\end{equation}   By the range of $l(i)$ and $h(i)$, we have $$\frac{l(i)n}{q^i+1}< l(i)q^{m-i} + h(i)<  \frac{l(i)n}{q^i-1},$$
i.e.
$$ \frac{(q^i-1)(\Delta_{3}+u)}{n} <l (i)<  \frac{(q^i+1)(\Delta_{3}+u)}{n}. $$
Thus,
it suffices to prove  that  there always exists $i\in\mathbb{Z}_m\backslash\{0\}$,  $ 1\leq l(i) \leq \frac{q^{i}-1}{2}$ such that 
\begin{equation}\label{equ5}
	\frac{(q^i-1)(\Delta_3+u)}{n} <l (i)<  \frac{(q^i+1)(\Delta_3+u)}{n}. 
\end{equation}
It is obvious that there exists $1\leq i \leq m-1$ such that $ \frac{(q^i-1)(\Delta_3+u)}{n}\geq1$. Moreover, 
\begin{align*}
	\frac{(q^i+1)(\Delta_3+u)}{n} < \frac{(q^i+1)\Delta_2}{n} &=\frac{q^i+1}{n}\cdot	\frac{(q-1)^2(q^2-1)n}{2(q^4+1)}=\frac{(q-1)^2(q^2-1)(q^i+1)}{2(q^4+1)} .
\end{align*}
By $\frac{q^i+1}{q^4+1} \leq \frac{q^i-1}{q^4-1}$, we have 
\begin{equation*}
	\frac{(q-1)^2(q^2-1)(q^i+1)}{2(q^4+1)} \leq 
	\frac{(q-1)^2(q^2-1)(q^i-1)}{2(q^4-1)} < \frac{q^i-1}{2}.
\end{equation*} 
Therefore, there always exists $1\leq i\leq m-1$, $ 1\leq l \leq \frac{q^{i}-1}{2}$   such that   (\ref{equ5}) holds. Hence $		\delta_3=\Delta_3$. By a similar analysis, we can prove that $\delta_4=\Delta_4$.
\end{proof}

\begin{lem}\label{lem13}
With the notations given before. Then $\left|C_{	\delta_2}\right|=8$ and $\left|C_{	\delta_3}\right|=\left|C_{	\delta_4}\right|=2m$.
\end{lem}
\begin{proof}
It is easy to check that $\left|C_{\delta_2}\right|=8$. Since  ${\rm Ord}_{n}(q)=2m$,  $\left|C_{\delta_3}\right|$ is a divisor of $2m$. We claim that  $\left|C_{		\delta_3}\right|=2m$. For otherwise,
by the assumption  $m\geq 12$ and  $m\equiv 4~ ({\rm mod}~ 8)$,  we get 
that $\nu_2(m)=2$ and $m=4(1+2i)$ for some $i\geq 1$ from Lemma \ref{lem00}. 
Suppose that $\left|C_{\delta_3}\right|=2^\epsilon i'$ where $\epsilon\leq 3$ and $i' \mid (1+2i)$. Note that if $\epsilon=3$, then $i' \neq 1+2i$ and $i'\leq \frac{1+2i}{3}$.
By the definition of cyclotomic coset, we have
$$ n \mid 	\delta_3(q^{2^\epsilon i'}-1).  $$
If $\epsilon \leq 2$, then  $ \frac{2^\epsilon i'}{\operatorname{gcd}(m, 2^\epsilon i')}=1$. By Lemma \ref{lem1},  one has $ n \mid 2	\delta_3$, which never holds since $\delta_3<\delta_1=\frac{n}{2}$. 	If $\epsilon=3$, then  $ \frac{2^\epsilon i'}{\operatorname{gcd}(m,  2^\epsilon i' )}$ is even. Again by Lemma \ref{lem1},  we have  ${\rm gcd}(n,  q^{8i'}-1)=q^{4i'}+1$ and $n\mid 		\delta_3(q^{4i'}+1)$.
Furthermore, 
\begin{align*}
	\delta_3(q^{4i'}+1)=  \frac{q^{4i'}+1}{q^4+1} \frac{(q-1)^2(q^2-1)}{2}(n-2(q^{m-8}+1)),
\end{align*}
this implies 
\begin{align}\label{eq123}
	n \mid  \frac{q^{4i'}+1}{q^4+1}(q-1)^2(q^2-1)(q^{m-8}+1).
\end{align}
By Lemma \ref{lem111}, we have ${\rm gcd}(n,q^{m-8}+1)=q^4+1$. Then (\ref{eq123}) becomes
\begin{align}\label{eq1234}
	n \mid (q^{4i'}+1)(q-1)^2(q^2-1),
\end{align}
which never holds since $4i'+4\leq 4\cdot\frac{1+2i}{3}+4<4(1+2i)=m$.
Thus  $\left|C_{\delta_3}\right|=2m$. By a similar analysis, we can prove that $\left|C_{\delta_4}\right|=2m$.
\end{proof}		

%

With the conclusions of Lemmas \ref{lem4} and \ref{lem13},  we determine the dimensions of some antiprimitive BCH codes as follows.
\begin{thm}\label{th3}
Let $q$ be an odd prime power,  $m\geq 12$ and  $m\equiv 4~ ({\rm mod}~ 8)$ and $n=q^m+1$. Let $\delta_1=\frac{n}{2}$, $\delta_2$, $\delta_3$ and  $\delta_4$ be given in Lemma \ref{lem4}. Then the LCD BCH code
$\mathcal{C}_{(q,q^m+1,\delta+1,0)}$ has parameters 
$$[q^m+1,8i-7,d\geq 2\delta]$$
if $	\delta_{i+1}+1                                                                                                                                                                                                                                                                                   \leq \delta \leq 	\delta_i$ (i=1,2);  
$\mathcal{C}_{(q,q^m+1,\delta+1,0)}$ has parameters 
$$[q^m+1,2m+9,d\geq 2\delta]$$
if $	\delta_4+1                                                                                                                                                                                                                                                                                  \leq \delta \leq \delta_3$;
and 
$\mathcal{C}_{(q,q^m+1,\delta_4+1, 0)}$ has parameters 
$$[q^m+1,4m+9,d\geq 2\delta_4].$$  
\end{thm}
We  employ SageMath to give  the following  example, where the dimensions are consistent with the results in Theorem \ref{th3}.
\begin{exa}
Let $q=3$ and $m=12$. Then the LCD BCH code $\mathcal{C}_{(3,531442,103697,0)}$ has the parameters $[531442, 9, d\geq 207392]$ and   $\mathcal{C}_{(3,531442,103665,0)}$ has the parameters   $[531442,33,d\geq 207328]$.
\end{exa}

\subsection{\texorpdfstring {$m=4$}{}}

The proofs of the following partial results are omitted because they  are essentially the same as that the one  given in the case where $ m\geq 12$ and  $m\equiv 4~ ({\rm mod}~ 8)$.
\begin{lem}\label{lem5}
Let  $q$ be an odd prime power,  $m=4$ and $n=q^4+1$. Let 
\begin{equation*}
	\begin{aligned}
		&	\Delta_{2,4}:=\frac{(q-1)^2(q^2-1)}{2},  \Delta_{3,4}:=\frac  {(q-1)^2(q^2-1)}{2}-(q-1) \text{ and }\\
		&  \Delta_{4,4}:=\frac  {(q-1)^2(q^2-1)}{2}-q. 
	\end{aligned}
\end{equation*}
Then $\Delta_{2,4}$, $\Delta_{3,4}$ and $\Delta_{4,4}$  are  coset leaders.
\end{lem}

\begin{lem}\label{lem6}
With the notations given above. Then 
\begin{equation*}
	\delta_2=\Delta_{2,4}, \delta_3=\Delta_{3,4} \text{ and } 
	\delta_4=\Delta_{4,4}
\end{equation*}
\end{lem}
\begin{proof}
Let $i=1$ and $l=\frac{q-1}{2}$ in Proposition \ref{prop1}. Then $s$ is not a coset leader if 
\begin{equation*}
	\frac{(q-1)n}{2(q+1)}<s<\frac{n}{2}.
\end{equation*}
Let $i=2$ and $l=\frac{(q-1)^2}{2}$ in Proposition \ref{prop1}. Then $s$ is not a coset leader if 
\begin{equation*}
	\frac{(q-1)^2n}{2(q^2+1)}<s<\frac{(q-1)n}{2(q+1)}.
\end{equation*}
Since $m$ is even, 	 $	\frac{(q-1)n}{2(q+1)}$ is not an integer. Moreover,
\begin{align*}
	\biggl\lfloor 	\frac{(q-1)^2n}{2(q^2+1)} \biggr\rfloor =  \Delta_{2,4}.	 
\end{align*}	
By Lemma \ref{lem5}, we have  $\delta_2=\Delta_{2,4}$. The remaining proof
is essentially the same as that    given in  the proof of
Lemma \ref{lem4}, and then is  omitted  here.
\end{proof}

\begin{lem}\label{lem9}
Let  $q$ be an odd prime power and $m=2^k$ with $k\geq2$. Then   $\left|C_{s}\right|=2m$ for any coset leader  $0<s<\frac{n}{2}$.
\end{lem}
\begin{proof}
It is known that ${\rm Ord}_{n}(q)=2m$, then $\left|C_{s}\right|$ is a divisor of $2^{k+1}$. Let  $\left|C_{s}\right|=2^t$ where $ 0\leq t\leq k+1$.  We claim that $t=k+1$. For otherwise, suppose that $t\leq k$.
Since $\frac{2^t}{{\rm gcd}(2^k,2^t)}=1$, by Lemma \ref{lem1}, we have ${\rm gcd}(q^{2^k}+1, ~q^{2^t}-1)=2$.  Then by the definition of cyclotomic coset, we have
$$ n \mid s(q^{2^t}-1),  $$
i.e.,
$$ n \mid 2s,  $$
which is impossible.  The desire conclusion then follows. 
\end{proof}

With  Lemmas \ref{lem6} and \ref{lem9},  we determine the dimensions of 	$\mathcal{C}_{(q,q^4+1,\delta+1,0)}$ as follows.
\begin{thm}\label{th1}
Let  $q$ be an odd prime power and  $n=q^4+1$. Let $\delta_1=\frac{n}{2}$, $\delta_2$, $\delta_3$ and $\delta_4$ be given in Lemma \ref{lem6}. Then the LCD BCH code
$\mathcal{C}_{(q,q^4+1,\delta+1,0)}$ has parameters 
$$[q^4+1,8i-7,d\geq 2\delta]$$
if $\delta_{i+1}+1\leq \delta \leq \delta_{i}$   (i=1,2,3);  and 
$\mathcal{C}_{(q,q^4+1,\delta_{4}+1, 0)}$ has parameters 
$$[q^4+1,25,d\geq 2\delta_{4}].$$  
\end{thm}

We  employ SageMath to give  the following  example, where the dimensions of codes are consistent with the results in Theorem \ref{th1}.
\begin{exa}
Let $q=3$ and $m=4$. Then the LCD BCH code  $\mathcal{C}_{(3,82,15,0)}$ has the parameters $[82,17,28]$ and   $\mathcal{C}_{(3,82,17,0)}$ has the parameters   $[82,9,44]$.
\end{exa}

\subsection{\texorpdfstring{$m=8$}{}}

The proofs of the following partial results are omitted because they  are essentially the same as that the one  given in the case where $ m\geq 12$ and  $m\equiv 4~ ({\rm mod}~ 8)$.
\begin{lem}\label{lem7}
Let $q$ be an odd prime power,  $m=8$ and $n=q^8+1$. Let 
\begin{align*}
	&\Delta_{2,8}=\frac{(q-1)^2(q^2-1)(q^4-1)}{2}, \Delta_{3,8}=\frac{(q-1)^2(q^2-1)(q^4-1)}{2}-(q-1)^2 ~\text{and}\\
	&\Delta_{4,8}=\frac{(q-1)^2(q^2-1)(q^4-1)}{2}-(q^3+q^2).
\end{align*} Then $\Delta_{2,8}$, $\Delta_{3,8}$ and  $\Delta_{4,8}$ are coset leaders.
\end{lem}

\begin{lem}\label{lem10}
With the notation given before. Then 
\begin{equation*}
	\delta_2=\Delta_{2,8}, \delta_{3}=\Delta_{3,8} ~\text{and}~\delta_{4}=\Delta_{4,8}.
\end{equation*}
\end{lem} 
\begin{proof}
Let $i=1$ and $l=\frac{q-1}{2}$ in Proposition \ref{prop1}. Then $s$ is not a coset leader if 
\begin{equation*}
	\frac{(q-1)n}{2(q+1)}<s<\frac{n}{2}.
\end{equation*}
Let $i=2$ and $l=\frac{(q-1)^2}{2}$ in Proposition \ref{prop1}. Then $s$ is not a coset leader if 
\begin{equation*}
	\frac{(q-1)^2n}{2(q^2+1)}<s<\frac{(q-1)n}{2(q+1)}.
\end{equation*}
Let $i=4$ and $l=\frac{(q-1)^2(q^2-1)}{2}$ in Proposition \ref{prop1}. Then $s$ is not a coset leader if 
\begin{equation*}
	\frac{(q-1)^2(q^2-1)n}{2(q^4+1)}<s<	\frac{(q-1)^2n}{2(q^2+1)}.
\end{equation*}
Since $m$ is even, 	$\frac{(q-1)^2n}{2(q^2+1)}$ and $	\frac{(q-1)n}{2(q+1)}$ are not integers. Moreover,
\begin{align*}
	\biggl\lfloor 	\frac{(q-1)^2(q^2-1)n}{2(q^4+1)} \biggr\rfloor =  	\Delta_{2,8}.
\end{align*}
By Lemma \ref{lem7}, we obtain that $\delta_{2}=\Delta_{2,8}$. 
The remaining proof
is essentially the same as that    given in  the proof of
Lemma \ref{lem4}, and then  is omitted  here.
\end{proof}

With  Lemmas \ref{lem9} and \ref{lem10},  we determine the dimension of 	$\mathcal{C}_{(q,q^8+1,\delta+1,0)}$ as follows.
\begin{thm}\label{th2}
Let  $q$ be an odd prime power and  $n=q^8+1$. Let $\delta_1=\frac{n}{2}$, $\delta_{2}$, $\delta_{3}$ and $\delta_{4}$ be given in Lemma \ref{lem10}. Then the LCD BCH code
$\mathcal{C}_{(q,q^8+1,\delta+1,0)}$ has parameters 
$$[q^8+1,16i-15,d\geq 2\delta]$$
if $\delta_{i+1}+1\leq \delta \leq \delta_{i}$   (i=1,2,3);  and 
$\mathcal{C}_{(q,q^8+1,\delta_{4}+1, 0)}$ has parameters 
$$[q^8+1,49,d\geq 2\delta_{4}].$$  
\end{thm}
We  employ SageMath to give  the following example, where the dimensions of codes are consistent with the results in Theorem \ref{th2}.
\begin{exa}
Let $q=3$ and $m=8$. Then the LCD BCH code  $\mathcal{C}_{(3,6562,1281,0)}$ has the parameters $[6562, 17, 4268]$ and   $\mathcal{C}_{(3,6562,1277,0)}$ has the parameters   $[6562,33,d\geq2552]$.
\end{exa}

\section{Dimensions of BCH codes of length \texorpdfstring {$n=\frac{q^m+1}{\lambda}$}{}}

Let $\lambda>1$ be a divisor of $q+1$. In this section, we will   determine the dimensions of LCD BCH codes $\mathcal{C}_{(q,n,\delta+1,0)}$ for $1\leq \delta-1 \leq \frac{ q^{\lfloor(m+1)/2\rfloor}}{\lambda}$. For some special values of $\lambda$, we also give the first few largest coset leaders  and then determine the dimensions of some LCD BCH code of length $n=\frac{q^m+1}{\lambda}$ with   large designed distances. The following lemma will play an important role in this
section.

\begin{lem}\label{lemm1}
For every integer $1\leq s \leq n-1$,  $s$ is (or is  not) a coset leader of $C_s$ modulo $n$ if and only if $\lambda s$ is (or is not) a coset leader of $C_{\lambda s}$ modulo $\lambda n$. Moreover, $\left|C_{ s}\right|=\left|C_{\lambda s}\right|$.
\end{lem}
\begin{proof}
We claim that $\lambda s$ is a coset leader modulo $\lambda n$ if $s$ is a coset leader modulo $n$. For otherwise, there exists an integer $1 \leq J<\lambda s$ such that
\begin{align}\label{equ00}
	\lambda s \equiv  Jq^l ~({\rm mod } ~\lambda n)	 
\end{align}
for some $l$. Then $J\neq \lambda j$ for $1\leq j <s$ because $s$ is a coset leader modulo $n$. Put $J=h+\lambda c$, where $h\in \{1,2,..,\lambda-1\}$ and  $c \in \{0,1,\cdots,s-1\}$. By (\ref{equ00}), we get 
$ \lambda \mid   \lambda s - (h+\lambda s )q^l$ and $\lambda\mid hq^l$, which is impossible. Hence, the above claim holds. If $\lambda s$ is a coset leader modulo  $\lambda n$, it is obvious that $s$ is a coset leader modulo $n$.  The partial result of `is not' can be proved by  a similar analysis. Let $l_1$ and $l_2 $ be the sizes of $\left|C_s\right|$ and $\left|C_{\lambda s}\right|$, respectively. Then  $n \mid s(q^{l_1}-1)$ and $\lambda n\mid 
\lambda s (q^{l_2}-1)$. By the definition of the size of cyclotomic coset,  we obtain $l_1=l_2$.  The result follows.
\end{proof}

Here, we list some results on the coset leaders modulo $q^m+1$. 
\begin{prop}\label{pro4}\cite{Li}
If $1\leq s \leq q^{\lfloor  \frac{m-1}{2} \rfloor}+1$ and $s \not\equiv 0~ ({\rm mod }~q)$, then $s$ is a coset leader with $\left|C_s\right|=2m$.
\end{prop}

\begin{prop}\cite{liu}
Let $m \geq 2$ be an even  integer. For $q^{\frac{m}{2}-1} \leq s \leq q^{\frac{m}{2}}$ with $s \not\equiv 0~ ({\rm mod} ~q) $, $s$ is a coset leader with $\left|C_s\right|=2m$.  
\end{prop}
\begin{prop}\cite{liu}\label{pro3}
Let $m\geq 3$ be an odd integer. For $q^{\frac{m-1}{2}} \leq s \leq q^{\frac{m+1}{2}}$ with $s \not\equiv 0~ ({\rm mod }~q)$, $s $ is a coset leader except  that $ s= q^{\frac{m+1}{2}}-r$ for $1\leq r \leq q-1$. Moreover,  $\left|C_s\right|=2$ if  $m=3$ and $s=q^2-q+1$;  $\left|C_s\right|=2m$, otherwise.
\end{prop}
\begin{rem}
Proposition \ref{pro3}  is imperfect in    \cite{liu} when $m=3$ 	because the authors ignored the  special case of   $s=q^2-q+1$. Here, we give its corrected result of it. 
\end{rem}

Below we aim to study the coset leaders modulo $n=\frac{q^m+1}{\lambda}$ in an interval  $[1, \frac{q^{\lfloor\frac{m+1}{2}\rfloor}}{\lambda}]$. When $\lambda=q+1$, the  results are studied in \cite{Yan}. Hence, we assume that $1 <\lambda<q+1$ and $\lambda\mid q+1$.

\begin{lem}\label{lem12}
Let $m\geq 3$ be an odd integer,  $1 <\lambda<q+1$ and $\lambda\mid q+1$.   For  $1\leq s \leq \frac{q^{\frac{m+1}{2}}}{\lambda}$ with $s \not\equiv 0 ({\rm mod }~q)$, 
$s$ is a coset leader except that 
\begin{equation*}
	s =
	\begin{cases}
		\frac{q^{\frac{m+1}{2}}+1}{\lambda}-c ~for~   1\leq c \leq \lfloor \frac{q}{\lambda}\rfloor	,& ~if~m \equiv 1~ ({\rm mod }~ 4),\\
		\frac{q^{\frac{m+1}{2}}-1}{\lambda}-c ~for~	0\leq c \leq \lfloor \frac{q-2}{\lambda}\rfloor,& ~~if~m \equiv 3 ~({\rm mod }~ 4).
	\end{cases}
\end{equation*} Moreover,  $\left|C_s\right|=2$ if $\lambda=m=3$, $q>3$ and $q \equiv 2 ~({\rm mod }~3)$ and $s=\frac{q^2-q+1}{\lambda}$;  $\left|C_s\right|=2m$, otherwise.
\end{lem}
\begin{proof}
For  $1 \leq \lambda s \leq q^{\frac{m+1}{2}}$, by Lemma \ref{lemm1}, Propositions \ref{pro4} and \ref{pro3},  $\lambda s $ is a coset leader modulo $q^m+1$  except $\lambda s= q^{\frac{m+1}{2}}-r$, i.e., $$s= \frac{q^{\frac{m+1}{2}}-(-1)^{{\frac{m+1}{2}}}}{\lambda} - \frac{r-(-1)^{{\frac{m+1}{2}}}}{\lambda},$$
for $1\leq r \leq q-1$.
Let $r=\lambda c + (-1)^{{\frac{m+1}{2}}}$. 
If $m \equiv 1~ ({\rm mod }~ 4)$,
then	$2\leq  \lambda c \leq q$, i.e., $1\leq c \leq \lfloor \frac{q}{\lambda}\rfloor$.
If $m \equiv 3~ ({\rm mod }~ 4)$, then $0\leq  \lambda c \leq q-2$, i.e., $0\leq c \leq \lfloor \frac{q-2}{\lambda}\rfloor$.
Again by  Lemma \ref{lemm1} and Proposition \ref{pro3}, if $m=3$ and $s=\frac{q^2-q+1}{\lambda}$ is an integer, then $s$ is a coset leader with $\left|C_s\right|=2$, which implies  $\lambda \mid q-2$.  Since $1<\lambda<q+1$ and $\lambda \mid q+1$, one  has $\lambda=3$, $q>3$ and  $q \equiv 2~ ({\rm mod }~3)$.
\end{proof}


With the results on the cyclotomic cosets in the interval $[1, \frac{q^{\frac{m+1}{2}}}{\lambda}]$ derived above, we have the following conclusions
on the parameters of BCH codes of   $n=\frac{q^m+1}{\lambda}$. The  proofs follow directly from Lemmas \ref{lem8} and  \ref{lem12}, and then details are omitted here.
\begin{thm}\label{th4}
Let $m\geq3$ be an odd integer,  $1 <\lambda<q+1$ and $\lambda\mid q+1$. 
For $1 \leq \delta-1 \leq \frac{q^{(m+1)/2}}{\lambda}$, 
if $m=\lambda=3$ , $q>3$ and $q \equiv 2~ ({\rm mod }~3)$, then
the LCD BCH code 	$\mathcal{C}_{(q,(q^3+1)/3,\delta+1,0)}$ has parameters $[\frac{q^3+1}{3},k,d\geq 2\delta]$
, where
\begin{equation*}
	k =
	\begin{cases}
		\frac{q^3+1}{3}-6(\delta-1-\lfloor \frac{\delta-1}{q}\rfloor)-1	,& ~if~\delta \leq \frac{q^2-q+1}{3},\\
		\frac{q^3+1}{3}-6(  \frac{q^2-q+1}{3}-1     -\lfloor \frac{\delta-1}{q}\rfloor)-3  	 ,& ~~if~\delta > \frac{q^2-q+1}{3};
	\end{cases}
\end{equation*}
Otherwise, 	the LCD BCH code 	$\mathcal{C}_{(q,(q^m+1)/\lambda,\delta+1,0)}$ has parameters $[\frac{q^m+1}{\lambda},k,d\geq 2\delta]$, where

\begin{equation*}
	k =
	\begin{cases}
		\frac{q^m+1}{\lambda}-2m(\delta-1-\lfloor \frac{\delta-1}{q}\rfloor)-1	,& ~if~\delta \leq \Delta ,\\
		\frac{q^m+1}{\lambda}-2m(\Delta-\lfloor \frac{\delta-1}{q}\rfloor)-1		  ,& ~if~\delta > \Delta ,\\
	\end{cases}
\end{equation*}
and
\begin{equation*}
	\Delta =
	\begin{cases}
		
		\frac{q^{\frac{m+1}{2}}+1}{\lambda}-\lfloor \frac{q}{\lambda}\rfloor	,& ~if~m \equiv 1~ ({\rm mod }~ 4),\\
		\frac{q^{\frac{m+1}{2}}-1}{\lambda}-\lfloor \frac{q-2}{\lambda}\rfloor	  ,& ~if~m \equiv 3~ ({\rm mod }~ 4).
	\end{cases}
\end{equation*}

\end{thm}

\begin{exa}
Let $q=5$,$m=3$ and $\lambda=3$. Then the LCD BCH code $\mathcal{C}_{(5,42,8,0)}$ has the parameters $[42, 11, 14]$ and   $\mathcal{C}_{(5,42,9,0)}$ has the parameters   $[42,9,22]$, which  the dimensions are consistent with the results in Theorem \ref{th4}.
\end{exa}

When $m$ is  even and $q$ is an odd prime power, the divisors of $q+1$ except 2  may not divide $q^m+1$. Thus, we only consider the case of $\lambda=2$.  We omit  the proof of  the following lemma because it is  
essentially the same as the one given in the proof of 
Lemma \ref{lem12}.

\begin{lem}
Let $q$ be an odd prime power and   $m\geq4$ be an even integer. For $1 \leq s \leq \frac{q^{m/2}}{2}$ with  $s \not\equiv 0~ ({\rm mod} ~q) $, $s$ is  a coset leader modulo $n=\frac{q^m+1}{2}$ with $\left|C_s\right|=2m$.
\end{lem}
\begin{thm}\label{thm5}
Let $q$ be an odd prime power,  $m\geq4$ be an even integer and $n=\frac{q^m+1}{2}$. 
For $1 \leq \delta-1 \leq \frac{q^{m/2}}{2}$, 
the LCD BCH code 	$\mathcal{C}_{(q,n,\delta+1,0)}$ has parameters $[n,n-2m(\delta-1-\lfloor \frac{\delta-1}{q}\rfloor)-1,d\geq 2\delta]$.
\end{thm}
\begin{exa}
Let $q=5$,$m=4$ and $\lambda=2$. Then the LCD BCH code $\mathcal{C}_{(5,313,11,0)}$ has the parameters $[313, 248, d\geq 20]$, which  the dimension is  consistent with the result in Theorem \ref{thm5}.
\end{exa}
In following, we determine  the first few largest
coset leaders modulo $n=\frac{q^m+1}{\lambda}$. Let $\delta_i^\lambda$ denote the $i$-th largest coset leader modulo $\frac{q^m+1}{\lambda}$.
In \cite{Yan},   the authors determined the   first  few largest coset leaders modulo $q^m+1$ for  an odd integer  $m$.
\begin{lem}\cite{Yan}
Let  $m$ be an odd positive integer and $n=q^m+1$. Then $\delta_1^1=\frac{q^m+1}{2}$, $\delta_2^1=\frac{(q-1)(q^m+1)}{2(q+1)}$ and $\delta_3^1=\frac{q-1}{2(q+1)}(q^m-2q^{m-2}-1)$. Moreover, $\left|C_{\delta_1^1}\right|=1$, $\left|C_{\delta_2^1}\right|=2$ and $\left|C_{\delta_3^1}\right|=2m$.
\end{lem}

From Lemma \ref{lemm1}, we can obtain some results.
\begin{coro}\label{coro1}
Let $m$ be an odd positive integer, $\lambda=2$, $q \equiv 1 ~({\rm mod } ~4)$  and $n=\frac{q^m+1}{\lambda}$. Then $\delta_1^2=\frac{(q-1)(q^m+1)}{4(q+1)}$ and $\delta_2^2=\frac{q-1}{4(q+1)}(q^m-2q^{m-2}-1)$. Moreover,  $\left|C_{\delta_1^2}\right|=2$ and $\left|C_{\delta_2^2}\right|=2m$.
\end{coro}
\begin{thm}\label{thm55}
Let $m$ be an odd positive integer, $\lambda=2$, $q \equiv 1 ~({\rm mod } ~4)$  and $n=\frac{q^m+1}{\lambda}$. Let $\delta_1^2$ and $\delta_2^2$  be given in Corollary \ref{coro1}. Then the LCD BCH code
$\mathcal{C}_{(q,n,\delta+1,0)}$ has parameters 
$$[n,2,d\geq 2\delta]$$
if $\delta_2^2+1\leq \delta \leq \delta_1^2$   and 
$\mathcal{C}_{(q,n,\delta_2^2+1,0)}$ has parameters 
$$[n,2+2m,d\geq 2\delta_2^2].$$  
\end{thm}
\begin{exa}
Let $q=5$,$m=3$ and $\lambda=2$. Then the LCD BCH code $\mathcal{C}_{(5,63,21,0)}$ has the parameters $[63, 2, d\geq 40]$ and   $\mathcal{C}_{(5,63,20,0)}$ has the parameters $[63, 8, d\geq 38]$, which  the dimensions are consistent with the results in Theorem \ref{thm55}.
\end{exa}

For an even integer $m$, we can obtain some analogous results. The following  corollary is obtained directly from Lemma 22 in \cite{liu3} and Lemma \ref{lemm1}.
\begin{coro}\label{coro2}
Let  $q$ be an odd prime power, $m$ be an even positive integer and $\lambda=2$. Then $\delta_1^2=\frac{\delta_2}{2}$ , where $\delta_2$ is defined by Theorem \ref{thm0}. Moreover, 
\begin{equation*}	 
	\left| C_{\delta_1^2}  \right|=\left\{\begin{array}{ll}2^{\nu_2(m)+1}, & if m\equiv 2^{\nu_2(m)}~ ({\rm mod}~ 2^{{\nu_2(m)}+1}), \\
		2m, &if m=2^k\geq2.\\
	\end{array}
	\right. 
\end{equation*}
\end{coro}

The following corollary  from  Lemmas \ref{lem4}, \ref{lem13} and  \ref{lemm1}.
\begin{coro}\label{coro3}
Let  $q$ be an odd prime power and $\lambda=2$. Let  $m\equiv 2^{\nu_2(m)}~({\rm mod}~ 2^{{\nu_2(m)}+1})$ and $m\geq 2^{\nu_2(m)}+2^{{\nu_2(m)}+1}$ for $\nu_2(m)=1,2$. Let $n=\frac{q^m+1}{2}$. Then $\delta_2^2=\frac{\delta_3}{2}$ and $\delta_3^2=\frac{\delta_4}{2}$, where $\delta_3$ and $\delta_4$ are defined by Conjecture \ref{con11} and Lemma \ref{lem4}. Moreover,  $\left|C_{	\delta_2^2}\right|=\left|C_{\delta_3^2}\right|=2m$.
\end{coro}
\begin{thm}\label{thm3}
Let  $q$ be an odd prime power and $\lambda=2$. Let  $m\equiv 2^{\nu}~ ({\rm mod}~ 2^{{\nu_2(m)}+1})$ and $m\geq 2^{\nu_2(m)}+2^{{\nu_2(m)}+1}$ for $\nu_2(m)=1,2$. Let $n=\frac{q^m+1}{2}$. Let $\delta_1^2$, $\delta_2^2$ and $\delta_3^2$  be given in Corollaries \ref{coro2} and \ref{coro3}. Then the LCD BCH code
$\mathcal{C}_{(q,n,\delta+1,0)}$ has parameters 
$$[n,2m(i-1)+2^{v+1},d\geq 2\delta]$$
if $\delta_{i+1}^2+1\leq \delta \leq \delta_i^2$ (i=1,2)  and 
$\mathcal{C}_{(q,n,\delta_3^2+1,0)}$ has parameters 
$$[n,4m+2^{v+1},d\geq 2\delta_3^2].$$  
\end{thm}
\begin{exa}
Let $q=3$,$m=6$ and $\lambda=2$. Then the LCD BCH code $\mathcal{C}_{(5,365,72,0)}$ has the parameters $[365, 16, d\geq 132]$ and  $\mathcal{C}_{(5,365,66,0)}$ has the parameters $[365, 28, d\geq 130]$, which  the dimensions are consistent with the results in Theorem \ref{thm3}.
\end{exa}

\section{Concluding remarks}

In this paper, we determine the first few largest coset leaders modulo $q^m+1$, where $q$ is an odd prime power and $m=4,8$, or  $m\equiv 4~({\rm mod } ~8 )$. Consequently, the dimensions of some LCD BCH codes of length $q^m+1$ with   large designed distances
are given.  We also determine the dimensions of the LCD BCH codes $\mathcal{C}_{(q,(q^m+1)/\lambda,\delta+1,0)}$  with $1\leq \delta-1 \leq \frac{ q^{\lfloor(m+1)/2\rfloor}}{\lambda}$  and $\lambda\mid q+1$. When $n=\frac{q^m+1}{\lambda}$, we determine the first few largest coset leaders modulo $n$ for some special values of $\lambda$ and thus determine the dimensions of some LCD BCH codes of length $n$ with   large designed distances.

Actually, when $m\equiv 2^{\nu(m)} ~({\rm mod}~2^{\nu(m)+1})$ and $m\geq 2^{\nu(m)}+2^{\nu(m)+1}$ for a fixed integer $\nu(m)\geq3$, we also can  verify that the statements for $\delta_3$ and $\delta_4$ in Conjecture \ref{con11} are true. For example, suppose that  $m\equiv 8~ ({\rm mod}~16)$ and $m\geq 24$. By   Conjecture \ref{con11}, we have 
\begin{align*}
\delta_3&=\frac{(q-1)^2(q^2-1)(q^4-1)(q^m-2q^{m-16}-1)}{2(q^8+1)}.
\end{align*} As is easily checked,  we have 
\begin{equation}\label{eq10}
\begin{aligned}
	\delta_3&=\frac{(q-1)^2(q^2-1)(q^4-1)(q^m-2q^{m-16}-1)}{2(q^8+1)}\\
	&=\frac{q^m-1}{2}-q^{m-1}+q^{m-3}-q^{m-4}+q^{m-5}-q^{m-7}+q^{m-9}-q^{m-11}
	+q^{m-12}\\
	&-q^{m-13}+q^{m-15}-q^{m-16}+q^{m-17}-q^{m-19}+q^{m-20}-q^{m-21}+q^{m-23}\\
	&+\sum_{j=1}^{(m-24)/16} (-q^{16j-1}+q^{16j-3}-q^{16j-4}+q^{16j-5}-q^{16j-7}+q^{16j-9}-q^{16j-11}\\
	&+q^{16j-12}-q^{16j-13}+q^{16j-15}).
\end{aligned}
\end{equation}
By the similar analyses with  Lemmas \ref{lem3}, \ref{lem11} and \ref{lem14}, we can prove that $\delta_3$ is the third-largest coset leader modulo $q^m+1$ when $m\equiv 8 ~({\rm mod}~16)$ and $m\geq 24$. However, for an unfixed integer $\nu(m)$, the expression of $\delta_3$ and $\delta_4$  that are similar in form to Eq. (\ref{eq10}) in Conjecture \ref{con11} are very complex, and then  it will be difficult to prove that  $\delta_3$ and $\delta_4$ in Conjecture \ref{con11} are  the third-largest and fourth-largest coset leaders, respectively.
Therefore, a possible direction for the future work is to improve our method and confirm  Conjecture \ref{con11} when $\nu(m)$ is unfixed.


\end{document}